\documentclass{fundam}

\usepackage[utf8]{inputenc}
\usepackage{amsfonts, amsbsy, amssymb,amsmath,graphicx}
\usepackage{caption}
\usepackage{cite}
\usepackage{subcaption}
\usepackage[labelformat=simple]{subcaption}

\usepackage{graphics}
\usepackage[english]{babel}
\usepackage{graphicx}
\usepackage{lineno}
\usepackage{enumerate}
\usepackage{tikz}
\usepackage{hyperref}
\usepackage{color}

\usetikzlibrary{calc,positioning}
\hypersetup{colorlinks=true}
\hypersetup{colorlinks=true, linkcolor=blue, citecolor=blue,urlcolor=blue}

\usepackage{changes}

\usepackage{hyperref}

\usepackage{enumitem}
\setlist[enumerate]{noitemsep,topsep=3pt}
\setlist[itemize]{noitemsep,topsep=3pt}

\newenvironment{proofofclaim}{\vspace{-6mm} \paragraph{\normalfont \textit{Proof of Claim \theclaim.}} \hspace{-2.5mm}}{\hfill$\blacklozenge$ \medskip}

\newcommand{\msrepr}{\ensuremath{\operatorname{m}}}

\newcommand{\msl}{\{\hspace*{-0.1cm}\{}
\newcommand{\msr}{\}\hspace*{-0.1cm}\}}

\usepackage{hyperref}

\begin{document}


\setcounter{page}{315}
\publyear{24}
\papernumber{2185}
\volume{191}
\issue{3-4}

\finalVersionForARXIV


\title{Complexity and Equivalency of  Multiset Dimension and ID-colorings}

\author{Anni Hakanen\thanks{Also affiliated at: Department of Mathematics and Statistics, University of Turku,
                 FI-20014, Finland}\thanks{Address for correspondence: Department of Mathematics and Statistics, 20014 Turun Yliopisto, Finland}
\\
Universit\'e Clermont-Auvergne, CNRS\\
Mines de Saint-\'Etienne, Clermont-Auvergne-INP\\
 LIMOS, 63000 Clermont-Ferrand, France \\
 anehak@utu.fi
 \and  Ismael G. Yero \\
 Department of Mathematics, Universidad de C\'{a}diz\\
 Av. Ram\'on Puyol s/n, 11202 Algeciras, Spain\\
 ismael.gonzalez@uca.es }

\runninghead{A. Hakanen and  I.G. Yero}{Complexity and Equivalency of  Multiset Dimension and ID-colorings}

\maketitle

\begin{abstract}
This investigation is firstly focused into showing that two metric parameters represent the same object in graph theory. That is, we prove that the multiset resolving sets and the ID-colorings of graphs are the same thing. We also consider some computational and combinatorial problems of the multiset dimension, or equivalently, the ID-number of graphs. We prove that the decision problem concerning finding the multiset dimension of graphs is NP-complete. We consider the multiset dimension of king grids and prove that it is bounded above by $4$. We also give a characterization of the strong product graphs with one factor being a complete graph, and whose multiset dimension is not infinite.
\end{abstract}

\begin{keywords} multiset resolving set; multiset dimension; ID-colorings; ID-number; king grids.\\
{\bf AMS Subj.\ Class.\ (2020)}: 05C12; 05C76.
\end{keywords}

\section{Introduction}

The concept of metric dimension in graphs is one of the classical parameters in the area of graph theory. It is understood it has been independently introduced in the decade of 1970 in the two separate works \cite{Harary} and \cite{Slater}, which were aimed to consider identification properties of vertices in a graph. These identification properties were also connected with the Mastermind game in \cite{Caceres-2007}, and problems related to pattern recognition and image processing in \cite{Melter-1984}.
 After these two seminal works, the research on the topic remained relatively quiet until the first years of the new century, where the number of articles on the topic exploded. From this point on, several theoretical and applied results have been appearing, and nowadays, the metric dimension of graphs is very well studied. It is not our goal to include a lot of references on this fact, and we simply suggest the interested reader to consult the two recent surveys \cite{dorota-2022+,till-2022+}, which have a fairly complete amount of information on metric dimension in graphs and related topics.

\medskip
Among the research lines addressed in the investigation with metric dimension in graphs, a remarkable one is that of considering variations of the classical concept that are giving more insight into the main concept, and are indeed of independent interest. The range of variations include generalizations of the concept, particularizations in the style of identification of vertices, union of the metric properties with other graph situations, identification of other elements (edges for instance) of graphs, etc. The survey \cite{dorota-2022+} contains information on several of the most important variations already known.

\smallskip
As it happens, sometimes some variations could relate much between them, and sometimes there could be some that indeed represent the same structure. This is one of the contributions of our work. We describe the equivalence of two metric concepts that, at a first glance, appear to be separate from each other. To this end, we first present some basic terminology and notations that shall be used in our article.

\smallskip
Along our work, all graphs are simple, undirected and connected, unless we will specifically state the contrary. Let $G$ be a graph with vertex set $V(G)$. Two vertices $u,v$ are identified (recognized or determined) by a vertex $w$ if $d_G(u,w)\ne d_G(v,w)$ where $d_G(x,y)$ stands for the distance (in its standard version) between $x$ and $y$. A set $S\subset V(G)$ is a \emph{resolving set} for $G$, if every two vertices $u,v\in V(G)$ are identified by a vertex of $S$. The cardinality of a smallest possible resolving set for $G$ is the \emph{metric dimension} of $G$, denoted $\dim(G)$. A resolving set of cardinality $\dim(G)$ is called a \emph{metric basis}. These concepts are from \cite{Harary} and \cite{Slater}, although in the latter work, they had different names. Some recent significant contributions on the metric dimension of graphs are for instance \cite{Claverol,Geneson,Mashkaria,Sedlar,Wu}.

\smallskip
The metric representation of a vertex $x$ with respect to an ordered set of vertices $S=\{v_1,\dots,v_k\}$ is the vector $r(x|S)=(d_G(x,v_1),\dots,d_G(x,v_k))$. It can be readily observed that a set $S$ of vertices is a resolving set of a graph $G$ if and only if the set of metric representations of vertices of $G$ are pairwise different. This terminology turns out to be very useful while working with this concept.

\medskip
Throughout our exposition, we are then focused on studying two (indeed one) metric dimension related parameters: the multiset dimension and the ID-number of graphs, which are variants of metric dimension that use multisets of distances instead of vectors, to uniquely identify the vertices of a graph. Our exposition is organized as follows. The next section is dedicated to prove that in fact the multiset dimension and the ID-number of graphs are the same parameter, and based on this, we just follow the terminology of multiset resolving sets and multiset dimension in our work. Section \ref{sec-complexity} is focused on showing that the decision problem concerning finding the multiset dimension of graphs is NP-complete. Next, in Section \ref{sec:king-grid}, we consider the multiset dimension of king grids, namely, the strong product of a path with itself. Section \ref{sec:strong-prod} contains a characterization of the strong product graphs with one factor being a complete graph, and whose multiset dimension is not infinite. Finally, we end our work with some questions and open problems that might be of interest for future research.

\section{Two equivalent metric concepts}

A modified version of the classical metric dimension of graphs was first presented in \cite{rino-2017} where the use of ``multisets'' instead of vectors was initiated while considering metric representations of vertices with respect to a given set. That is, if $u\in V(G)$ and $W = \{w_1, \ldots, w_t\}$, then the {\em multiset
representation of $u$ with respect to} $W$ is given by
$${\rm m}_G(u|W)=\msl d_G(u, w_1), \ldots, d_G(u, w_t) \msr,$$
where $\msl \cdot \msr$ represents a multiset. In order to facilitate our exposition, we write $\msl . \msr + i$ to denote the multiset obtained from $\msl . \msr$ by adding $i$ to every element of such multiset. The set $W$ is a \emph{multiset resolving set} for $G$ if the collection of multisets ${\rm m}_G(u|S)$ with $u\in V(G)$ are pairwise distinct. The \emph{multiset dimension} of $G$, denoted $\dim_{ms}(G)$, represents the cardinality of a smallest possible multiset resolving set of $G$. Multiset resolving sets do not always exist in a given graph. For those graphs $G$ which do not contain any multiset resolving set, the agreement that $\dim_{ms}(G)=\infty$ was taken in \cite{rino-2017}. Some other investigations on the multiset dimension of graphs are \cite{Alfarisi,bong-2021,Isariyapalakul,Khemmani-2020,Khemmani-2018}. It is clear that any multiset resolving set is also a resolving set, since the fact that two multisets are different implies that the vectors with the same elements are also different. This means that for any graph $G$,
\begin{equation}\label{eq-dim-mdim}
\dim(G)\le \dim_{ms}(G).
\end{equation}

On the other hand, the following concepts were defined in \cite{Chartrand-2021}. Consider a connected graph $G$ of diameter $d$ and a set of vertices $S\subset V(G)$. Now, for every vertex $x\in V(G)$, the \emph{code} of $x$ with respect to $S$ is the $d$-vector $\vec{d}(x|S)=(a_1,a_2,\dots,a_d)$ where $a_i$, with $i\in\{1,\dots,d\}$ represents the number of vertices in $S$ at distance $i$ from $x$. If all the codes of vertices of $G$ are pairwise different, then $S$ is called an \emph{identification coloring} or \emph{ID-coloring}. Moreover, a graph $G$ that has an ID-coloring is called an ID-graph. In this sense, for any ID-graph $G$, the cardinality of a smallest ID-coloring is the \emph{ID-number} of $G$, denoted by $ID(G)$. Other contributions in this direction are \cite{Kono-2022,Kono-2022a,Kono-2021}.

\medskip
We next show that the two parameters defined above are indeed the same. To this end, we may remark that for an ID-graph $G$ of diameter $d$ with an ID-coloring $S$, any vertex $x\in V(G)$ with code $\vec{d}(x|S)=(a_1,a_2,\dots,a_d)$ satisfies $\sum_{i=1}^{d}a_i=|S|$ if $x\notin S$ and $\sum_{i=1}^{d}a_i=|S|-1$ if $x\in S$.

\begin{theorem}
\label{th:equivalence}
Let $G$ be a graph of diameter $d$. Then $S\subset V(G)$ is an ID-coloring for $G$ if and only if $S$ is a multiset resolving set for $G$.
\end{theorem}

\begin{proof}
The result follows directly from the following fact. Let $S\subset V(G)$ and let $x\in V(G)$. Consider the multiset representation ${\rm m}_G(x|S)=\msl d_G(x, w_1), \ldots, d_G(x, w_t) \msr$. From ${\rm m}_G(x|S)$ we construct the vector $\vec{d}(x|S)=(a_1,a_2,\dots,a_d)$ as follows. Each $a_i$ equals the number of elements in ${\rm m}_G(x|S)$ with value $i$ for any $i\in\{1,\dots,d\}$. On the contrary, if we have the vector $\vec{d}(x|S)=(a_1,a_2,\dots,a_d)$, then the multiset representation ${\rm m}_G(x|S)$ can be obtained in the following way. First, if $\sum_{i=1}^{d}a_i=|S|-1$, then $x\in S$ and we need to add the value zero (0) to ${\rm m}_G(x|S)$. Otherwise, if $\sum_{i=1}^{d}a_i=|S|$, then $x\notin S$, and the value zero (0) does not belong to ${\rm m}_G(x|S)$. Now, in both cases, for every $a_i\in \vec{d}(x|S)$ such that $a_i\ne 0$, we add to ${\rm m}_G(x|S)$ $a_i$ elements equal to $i$.

Based on the equivalence that exists between these two representations, the fact that the codes of vertices of $G$ are pairwise different implies that the collection of multisets are pairwise distinct as well, and vice versa. Consequently, it is clear that a given set $S\subset V(G)$ is an ID-coloring of $G$ if and only $S$ is a multiset resolving set.
\end{proof}

Based on the equivalence above, we conclude the next consequence.
\begin{corollary}
\label{cor:equivalence}
For any graph $G$, $\dim_{ms}(G)=ID(G)$.
\end{corollary}

It is then now clear that graphs defined in \cite{Chartrand-2021} as ID-graphs are those ones satisfying that $\dim_{ms}(G)<\infty$ according to the terminology from \cite{rino-2017}.

\section{Complexity results}
\label{sec-complexity}

This section is centered into considering the proof of the NP-completeness for the following decision problem, which in addition allows to conclude that computing the multiset dimension of graphs is NP-hard, and clearly, based on Theorem \ref{th:equivalence} and Corollary \ref{cor:equivalence}, it means that computing the ID-number of graphs is NP-hard as well. Our proof is somewhat inspired by the proof of the NP-hardness of the outer multiset dimension problem presented in \cite{gil-pons-2019}.

\medskip
Distinct vertices $u$ and $v$ are called \emph{twins} if they have the same set of neighbors. It is well known that if two vertices are twins, then each (multiset) resolving set contains at least one of them.\\

\indent \textsc{Multiset Dimension} \\
\indent \textbf{Instance:} A graph $G=(V,E)$ and an integer $k$ satisfying $1 \le k \le |V|$. \\
\indent \textbf{Question:} Is $\dim_{ms}(G) \le k$?

\begin{theorem}
\label{th:complexity}
The \textsc{Multiset Dimension} problem is NP-complete.
\end{theorem}

\begin{proof}
	The problem is clearly in NP. We prove the NP-completeness by a reduction
	from 3-SAT. Consider an arbitrary input to 3-SAT, that is, a formula $F$ with $n$ variables and $m$ clauses, which does not have a clause containing both the positive and the negative literals of the same variable. Let $x_1, x_2, \dots, x_n$ be the variables, and
	let $C_1, C_2, \dots, C_m$ be the clauses of $F$. We next construct a connected
	graph $G$ based on this formula $F$. To this end, we use the following gadgets.
	
	For each variable $x_i$ we construct a \emph{variable gadget} as follows (see Figure
	\ref{fig:gadgetxi}).
	\begin{itemize}
		\item Vertices $T_i$, $F_i$ are the ``true'' and ``false'' ends of the gadget. The
		gadget is attached to the rest of the graph only through these vertices.
		\item Vertices $a_i^1$, $a_i^2$, $b_i^1$, $b_i^2$ represent the value of the
		variable $x_i$, that is, $a_i^1$ and $a_i^2$ will be used to represent that
		variable $x_i$ is true, and $b_i^1$ and $b_i^2$ that it is false. The vertices $a_i^1$ and $b_i^1$ are adjacent, and so are the vertices $a_i^2$ and $b_i^2$. Additionally, the vertices $a_i^1$ and $a_i^2$ are adjacent to $T_i$ and the vertices $b_i^1$ and $b_i^2$ are adjacent to $F_i$.
		\item $P_{t_i}=d_i^1d_i^2\cdots d_i^{t_i}$ is a path such that $d_i^1$ is adjacent to $T_i$ and $F_i$, while $d_i^{t_i}$ is adjacent to the two vertices $e_i^1$ and $e_i^2$. Notice that the vertices $e_i^1$ and $e_i^2$  are twins, and so, each multiset resolving set of $G$ must contain at least one of them.
	\end{itemize}
	
	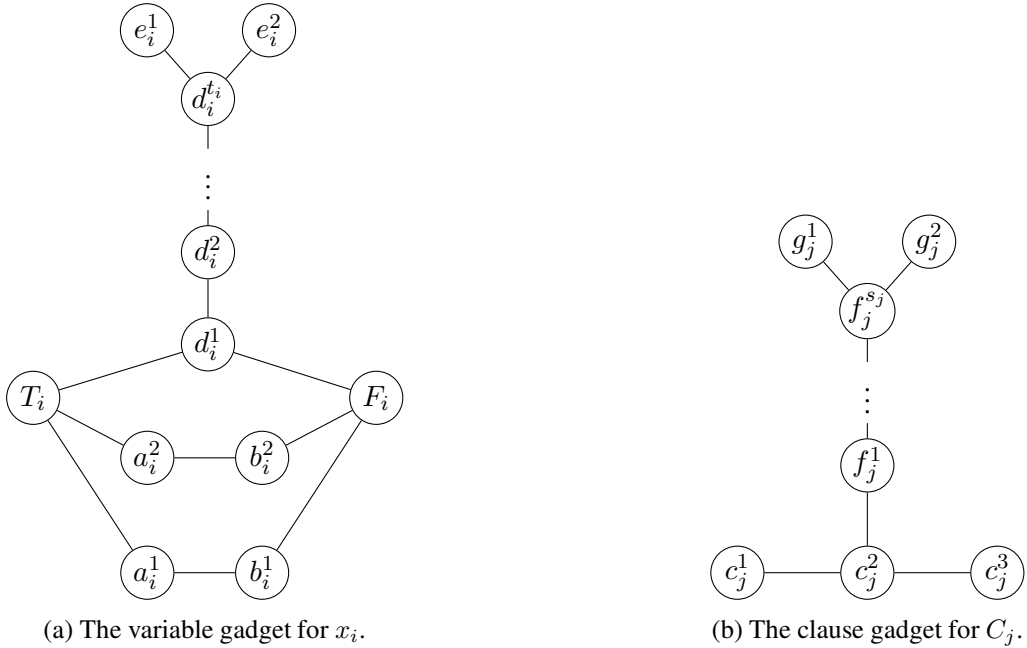
\begin{figure}
		\begin{subfigure}[b]{.44\linewidth}
			\centering
			\begin{tikzpicture}
				\tikzset{main node/.style={circle,fill=white,draw,minimum
						size=0.7cm,inner sep=0pt},}
				\node[main node] (1) {$a_i^1$};
				\node[main node] (2) [right = .8cm of 1] {$b_i^1$};
				\node[main node] (3) [above = .8cm of 1]  {$a_i^2$};
				\node[main node] (4) [above = .8cm of 2]  {$b_i^2$};
				\node[main node] (5) [above left = 1.8cm and 1cm of 1] {$T_i$};
				\node[main node] (6) [above right = 1.8cm and 1cm of 2] {$F_i$};
				\node[main node] (7) [above right = 0.2cm and 1.8cm of 5] {$d_i^1$};
				\node[main node] (8) [above = 0.5cm of 7] {$d_i^2$};
				\node[] (9) [above = 0.2cm of 8] {$\vdots$};
				\node[main node] (10) [above = 0.3cm of 9] {$d_i^{t_i}$};
				\node[main node] (11) [above left = 0.4cm and 0.3cm of 10] {$e_i^{1}$};
				\node[main node] (12) [above right = 0.4cm and 0.3cm of 10] {$e_i^{2}$};
				
				\path[draw]
				(1) edge node {} (2)
				(3) edge node {} (4)
				(5) edge node {} (1)
				(2) edge node {} (6)
				(4) edge node {} (6)
				(3) edge node {} (5)
				(5) edge node {} (7)
				(6) edge node {} (7)
				(7) edge node {} (8)
				(8) edge node {} (9)
				(9) edge node {} (10)
				(11) edge node {} (10)
				(12) edge node {} (10);
			\end{tikzpicture}
			\caption{The variable gadget for $x_i$.}
			\label{fig:gadgetxi}
		\end{subfigure}
		\hfill
		\begin{subfigure}[b]{.44\linewidth}
			\centering
			\begin{tikzpicture}
				\tikzset{main node/.style={circle,fill=white,draw,minimum size=0.7cm,inner
						sep=0pt},}
				\tikzset{many node/.style={circle,fill=white,draw,minimum size=0.9cm},}
				\node[main node] (1) {$c_j^1$};
				\node[main node] (2) [right = 1cm of 1]  {$c_j^2$};
				\node[main node] (3) [right = 1cm of 2]  {$c_j^3$};
				\node[main node] (4) [above = 0.7cm of 2] {$f_j^1$};
				\node[] (5) [above = 0.2cm of 4] {$\vdots$};
				\node[main node] (6) [above = 0.3cm of 5] {$f_j^{s_j}$};
				\node[main node] (7) [above left = 0.4cm and 0.3cm of 6] {$g_j^{1}$};
				\node[main node] (8) [above right = 0.4cm and 0.3cm of 6] {$g_j^{2}$};
				
				\path[draw]
				(1) edge node {} (2)
				(2) edge node {} (3)
				(2) edge node {} (4)
				(4) edge node {} (5)
				(5) edge node {} (6)
				(6) edge node {} (7)
				(6) edge node {} (8);
			\end{tikzpicture}
			\caption{The clause gadget for $C_j$.}
			\label{fig:gadgetCj}
		\end{subfigure}\vspace*{-2mm}
		\caption{The variables and clause gadgets used in the reduction.}\vspace*{-2mm}
	\end{figure}
	
	For each clause $C_j$ we construct a \emph{clause gadget} as follows (see Figure
	\ref{fig:gadgetCj}).
	
	\begin{itemize}
		\item The vertices $c_j^1$, $c_j^2$ and $c_j^3$ form a path. The vertices $c_j^1$ and $c_j^3$ will be helpful in determining whether the clause $C_j$ is satisfied.
		\item $P_{s_j}=f_j^1f_j^2\cdots f_{j}^{s_j}$ is a path such that $f_j^1$ is adjacent to $c_j^2$, while $f_{j}^{s_j}$ is adjacent to the two vertices $g_j^1$ and $g_j^2$. The vertices $g_j^1$ and $g_j^2$  are twins, and so, each multiset resolving set of $G$ must contain at least one of them.
	\end{itemize}

	The orders of the paths $P_{t_i}$ and $P_{s_j}$ must be pairwise distinct in order for the multiset representations to be different and our reduction to work. Let $t_i = 5(i+1)$ for all $i \in \{1, \ldots , n\}$ and $s_j = 5(n+j+1)$ for all $j \in \{1, \ldots , m\}$. The sum of all the vertices in the variable and clause gadgets is clearly polynomial in terms of $n+m$.
	
	The variable and clause gadgets are connected in the following way in order to construct our graph~$G$.
		\begin{itemize}
\itemsep=0.95pt
		\item Vertices $c_j^1$ are adjacent to vertices $T_i$, $F_i$ for all $j$ and $i$.
		\item If a variable $x_i$ does not appear in a clause $C_j$, then the vertices $T_i$, $F_i$ are adjacent to $c_j^3$.
		\item If a variable $x_i$ appears as a positive literal in a clause $C_j$, then the vertex $F_i$ is adjacent to $c_j^3$.
		\item If a variable $x_i$ appears as a negative literal in a clause $C_j$, then the vertex $T_i$ is adjacent to $c_j^3$.
	\end{itemize}
	
	Observe that $G$ is connected and its order is polynomial in terms of the quantity of variables and clauses of the 3-SAT instance. We shall show that $F$ is satisfiable if and only if $\dim_{ms}(G)=2m+n$. To this end, we proceed with a series of claims that will complete our whole reduction.
	
	\begin{claim}\label{claim:2n+m}
		We have $\dim_{ms}(G)\ge 2n+m$.
	\end{claim}
	\begin{proofofclaim}
		Let $S$ be a multiset basis of $G$.
		The vertices $e_i^1$ and $e_i^2$ are twins, and thus $S$ contains at least one of them for every $i\in\{1,\dots,n\}$. Similarly, the vertices $g_j^1$ and $g_j^2$ are twins and at least one of them is in $S$ for every $j\in\{1,\dots,m\}$. Finally, in order to have distinct multiset representation for the vertices $a_i^1$, $a_i^2$, $b_i^1$ and $b_i^2$, at least one of these four vertices must be in $S$. Thus, we have $\dim_{ms}(G)=|S|\ge 2n+m$.
	\end{proofofclaim}\vspace*{-1mm}

	\begin{claim}\label{claim:SATtoDIM}
		If $F$ is satisfiable, then $\dim_{ms} (G) = 2n+m$.
	\end{claim}
	\begin{proofofclaim}
		Consider a satisfying assignment for $F$ and construct a set $S^*$ containing $2n+m$ vertices as next described.
		\begin{itemize}
    \itemsep=0.95pt
			\item For each $i\in \{1,\dots, n\}$, we add the vertex $e_i^1$ to $S^*$.
			\item For each $j\in \{1,\dots, m\}$, we add the vertex $g_j^1$ to $S^*$.
			\item For each variable $x_i$, if $x_i = \texttt{true}$, then we add the vertex $a_i^1$ to $S^*$, otherwise if $x_i = \texttt{false}$, then we add the vertex $b_i^1$ to $S^*$.
		\end{itemize}
		
		We will show that the set $S^*$ is a multiset resolving set of $G$. We denote by $S^*_{x_i}$ and $S^*_{C_j}$ the vertices of the set $S^*$ that are \emph{not} in the gadget of $x_i$ and $C_j$, respectively. We will first express the multiset representations of the vertices in the variable gadgets with the help of the multiset representation of $T_i$ with respect to $S^*_{x_i}$. Since the vertices $c_j^1$ are adjacent to all $T_i$ and $F_i$, the distance from $T_i$ to all $T_{i'}$ and $F_{i'}$ for $i' \neq i$ is 2. Now, we have
		\begin{align*}
			\msrepr( T_i | S^*_{x_i} ) &= \msl
			& 3, &&|\;\;& \text{for each } i' \in \{1,\ldots,n\}, i' \neq i \\
			&& t_{i'}+3, &&|\;\;& \text{for each } i' \in \{1,\ldots,n\}, i' \neq i \\
			&& s_{j}+3, &&|\;\;& \text{for each } j \in \{1, \ldots , m\}
			\msr .
		\end{align*}
		For instance, the distance between the vertex $T_i$ and $e_{i'}^1$ equals $t_{i'}+3$, because one shortest path between them is $T_ic_1^1T_{i'}d_{i'}^1 \cdots d_{i'}^{t_{i'}}e_{i'}^1$. Other cases are deduced similarly from the construction of the graph $G$.

\medskip		
		Now, we can write the multiset representations of the vertices in the variable gadget of $x_i$ as follows.
		\begin{itemize}
			\item $\msrepr( T_i | S^* ) = \msrepr( T_i | S^*_{x_i} ) \cup \msl t_i + 1, y \msr$, where $y=1$ when $a_i^1 \in S^*$ and $y=2$ when $b_i^1 \in
                     S^*$.
			\item $\msrepr( F_i | S^* ) = \msrepr( T_i | S^*_{x_i} ) \cup \msl t_i + 1, y \msr$, where $y=1$ when $b_i^1 \in S^*$ and $y=2$ when $a_i^1 \in
                     S^*$.
\eject
			\item $\msrepr( d_i^h | S^* ) = (\msrepr( T_i | S^*_{x_i} ) + h ) \cup \msl t_i - h + 1, h + 1 \msr$ for all $h \in \{1, \ldots , t_i\}$.
			\item $\msrepr( e_i^h | S^* ) = (\msrepr( T_i | S^*_{x_i} ) + t_i + 1 ) \cup \msl t_i + 2, y \msr$, where $y = 0$ when $e_i^h = e_i^1$ and $y = 2$ when $e_i^h = e_i^2$.
			\item $\msrepr( a_i^h | S^* ) = (\msrepr( T_i | S^*_{x_i} ) + 1 ) \cup \msl t_i + 2, y \msr$, where when $a_i^1 \in S^*$, we have $y=0$ for $a_i^1$ and $y=2$ for $a_i^2$, and when $b_i^1 \in S^*$, we have $y=1$ for $a_1^i$ and $y=3$ for $a_i^2$.
			\item $\msrepr( b_i^h | S^* ) = (\msrepr( T_i | S^*_{x_i} ) + 1 ) \cup \msl t_i + 2, y \msr$, where when $a_i^1 \in S^*$, we have $y=1$ for $b_i^1$ and $y=3$ for $b_i^2$, and when $b_i^1 \in S^*$, we have $y=0$ for $b_1^i$ and $y=2$ for $b_i^2$. \vspace*{1mm}
		\end{itemize}
		
		As for the vertices of the clause gadgets, we will express them with an auxiliary representation as well. To that end, observe that
		\begin{align*}
			\msrepr( c_j^1 | S^*_{C_j} ) &= \msl
			& 2, 		&&|\;\;& \text{for each } i \in \{1,\ldots,n\} \\
			&& t_{i}+2, 	&&|\;\;& \text{for each } i \in \{1,\ldots,n\} \\
			&& s_{j'}+4,	&&|\;\;& \text{for each } j' \in \{1, \ldots , m\},
			j' \neq j \msr .
		\end{align*}
		
		We then write the multiset representations of the vertices of the clause gadget of $C_j$ other than $c_j^3$ as follows.
		\begin{itemize}
			\item $\msrepr( c_j^1 | S^* ) = \msrepr( c_j^1 | S^*_{C_j} ) \cup \msl s_j + 2 \msr$.
			\item $\msrepr( c_j^2 | S^* ) = (\msrepr( c_j^1 | S^*_{C_j} ) + 1) \cup \msl s_j + 1 \msr$.
			\item $\msrepr( f_j^h | S^* ) = (\msrepr( c_j^1 | S^*_{C_j} ) + h + 1) \cup \msl s_j - h + 1 \msr$ for all $h \in \{1, \ldots , s_j\}$.
			\item $\msrepr( g_j^h | S^* ) = (\msrepr( c_j^1 | S^*_{C_j} ) + s_j + 2) \cup \msl y \msr$, where $y = 0$ when $g_j^h = g_j^1$ and $y = 2$ when $g_j^h = g_j^2$. \vspace*{1mm}
		\end{itemize}
		
		Since the values of $t_i$ and $s_j$ are large and distinct enough (recall that $t_i = 5(i+1)$ and $s_j=5(n+j+1)$), the multiset representations of the vertices in the variable and clause gadgets (other than $c_j^3$) are pairwise distinct. Indeed, notice that the values $t_{i'}+3$ and $s_j+3$ form a pattern to the multiset representations of a vertex of a variable gadget which is easy to distinguish from the corresponding pattern of $t_i+2$ and $s_{j'}+4$ of a vertex from a clause gadget. Thus, we readily observe that the multiset representations of the vertices in a variable gadget are distinct from those of the vertices of clause gadgets. These patterns, or the anomalies present in them, to be more precise, are also the reason why vertices in two different variable gadgets (or clause gadgets) have distinct multiset representations. Indeed, there is a ``gap'' in this pattern where $t_i+3$ should be for all vertices of the gadget of $x_i$. Thus, the multiset representations of the vertices within the same variable or clause gadget are distinct.

\medskip		
		Let us then consider the vertices $c_j^3$. Similarly to the vertex $c_j^1$, the multiset representation of $c_j^3$ contains $s_j+2$ once, $t_i+2$ for each $i \in \{1, \ldots ,n\}$, and $s_{j'} + 2$ for each $j' \in \{1, \ldots , m\}$, $j' \neq j$. However, the distance from $c_j^3$ to the vertices $a_i^1$ and $b_i^1$ is 2 or 3 depending on whether the variable $x_i$ appears in the clause $C_j$ and in which form (positive or negative), and whether $x_i = \texttt{true}$ or $x_i = \texttt{false}$ according to the truth assignment. More precisely, if the variable $x_i$ does not appear in the clause $C_j$, then $d(c_j^3,a_i^1) = d(c_j^3,b_i^1) = 2$. If $x_i$ appears in $C_j$ but the clause $C_j$ is not satisfied due to $x_i$ (note that here $C_j$ can be satisfied but due to the truth value of some other variable $x_{i'}$), then the distance from $c_i^3$ to whichever of $a_i^1$ and $b_i^1$ is in $S^*$ is again 2. However, if $C_j$ is satisfied due to $x_i$, the distance from $c_j^3$ to whichever of $a_i^1$ and $b_i^1$ is in $S^*$ is 3. (For example, if $x_i$ appears as a positive literal in $C_j$, the edge $c_j^3F_i$ is present in $G$ whereas $c_j^3T_i$ is not. Thus, we have $d(c_j^3,a_i^1) = 3$ and $d(c_j^3,b_i^1) = 2$. Now, if the truth assignment of $x_i$ leads to $C_j$ being satisfied, we have $x_i = \texttt{true}$ and $a_i^1 \in S^*$. Therefore, the multiset representation of $c_j^3$ contains 3 due to $a_i^1$.)
		
\medskip
		Since the set $S^*$ is constructed using a truth assignment that satisfies $F$, there is at least one $3$ in the multiset representation of $c_j^3$. Thus, the multiset representation of $c_j^3$ is almost the same as the multiset representation of $c_j^1$ except that at least one 2 (in $\msrepr (c_j^1 | S^*)$) is swapped to 3. Thus, $c_j^3$ and $c_j^1$ have distinct multiset representations. Furthermore, based on the arguments concerning the multiset representations of the other vertices of $G$, it is clear that each $c_j^3$ has a distinct multiset representation compared to all other vertices of $G$.
		
		Consequently, the set $S^*$ is a multiset resolving set of $G$, and the claim holds due to Claim~\ref{claim:2n+m}.
	\end{proofofclaim}\vspace*{-1mm}

	\begin{claim}\label{claim:DIMtoSAT}
		If $\dim_{ms} (G) = 2n+m$, then $F$ is satisfiable.
	\end{claim}
	\begin{proofofclaim}
		Let $S$ be a multiset basis of $G$. By the arguments in the proof of Claim \ref{claim:2n+m}, the set $S$ must contain exactly one of the two vertices $e_i^1$ or $e_i^2$ for every $i\in \{1,\dots, n\}$; exactly one of the two vertices $g_j^1$ or $g_j^2$ for every $j\in \{1,\dots, m\}$; and exactly one of the vertices $a_i^1$, $a_i^2$, $b_i^1$ or $b_i^2$ for every $i\in \{1,\dots, n\}$.  By the same arguments as at the end of the proof of Claim~\ref{claim:SATtoDIM}, the vertices $c_j^1$ and $c_j^3$ must have distinct multiset representations due to some $a_i^h$ or $b_i^h$. Thus, the truth assignment where $x_i = \texttt{true}$ if $a_i^1$ or $a_i^2$ is in $S$, and $x_i = \texttt{false}$ if $b_i^1$ or $b_i^2$ is in $S$ for all $i \in \{1, \ldots , n\}$ satisfies$\;F$.
	\end{proofofclaim}

	This completes the reduction from 3-SAT to the \textsc{Multiset Dimension} problem.
\end{proof}

\section{The king grid}\label{sec:king-grid}

Based on the NP-completeness reduction made in the proof of Theorem \ref{th:complexity}, it is then desirable to consider the multiset dimension (or ID-number) of some non-trivial families of graphs. In connection with this, in this section we consider the strong product of a path $P_n$ with itself, also known as the king grid.

\medskip
The graph $G \boxtimes H$ is the \emph{strong product} of $G$ and $H$. The vertex set of $G \boxtimes H$ is the set $V(G) \times V(H) = \{ (u,v) \, | \, u \in V(G), v \in V(H) \}$. Two vertices $(g,h), (g',h') \in V(G \boxtimes H)$ are adjacent if $g = g'$ and $h$ is adjacent to $h'$ in $H$; or $g$ is adjacent to $g'$ in $G$ and $h = h'$; or $g$ is adjacent to $g'$ in $G$ and $h$ is adjacent to $h'$ in $H$. We write $V(P_n) = \{1,\ldots, n\}$ so that the vertices of $P_n \boxtimes P_n$ correspond to the coordinates of the $\mathbb{Z}^2$ lattice. Moreover, for a vertex $(i,j)\in V(P_n \boxtimes P_n)$ and an integer $q\ge 1$, by $D_q(i,j)$ we mean the set of vertices in $V(P_n \boxtimes P_n)$ at distance $q$ from $(i,j)$. Notice that such set $D_q(i,j)$ represents a kind of (not necessarily whole) ``border'' of a subgraph of $P_n \boxtimes P_n$ isomorphic to the strong product of two paths. See Figure \ref{fig:P6P6-sets-D} for two representative examples.

\begin{figure}[ht]
\vspace*{-1mm}
	\centering
	\begin{tikzpicture}[scale=0.85]
		\def\size{\footnotesize}
		\draw \foreach \x in {1,...,5} \foreach \y in {1,...,6} {
			(\x,\y) -- (\x+1,\y)
		};
		\draw \foreach \x in {1,...,6} \foreach \y in {1,...,5} {
			(\x,\y) -- (\x,\y+1)
		};
		\draw \foreach \x in {1,...,5} \foreach \y in {1,...,5} {
			(\x,\y) -- (\x+1,\y+1)
			(\x,\y+1) -- (\x+1,\y)
		};
		\draw \foreach \x in {1,...,6} \foreach \y in {1,...,6}{
			node[circle, draw, fill=white, inner sep=0pt, minimum width=7pt] (\x\y) at (\x,\y) {}
		};
		\draw \foreach \x in {(14),(24),(34),(44),(41),(42),(43)}{
			\x node[circle, draw, fill=black, inner sep=0pt, minimum width=7pt] {}
		};
        \draw \foreach \x in {(22)}{
			\x node[circle, draw, fill=red, inner sep=0pt, minimum width=7pt] {}
		};
	\end{tikzpicture}
\hspace*{1.6cm}
	\begin{tikzpicture}[scale=0.85]
		\def\size{\footnotesize}
		\draw \foreach \x in {1,...,5} \foreach \y in {1,...,6} {
			(\x,\y) -- (\x+1,\y)
		};
		\draw \foreach \x in {1,...,6} \foreach \y in {1,...,5} {
			(\x,\y) -- (\x,\y+1)
		};
		\draw \foreach \x in {1,...,5} \foreach \y in {1,...,5} {
			(\x,\y) -- (\x+1,\y+1)
			(\x,\y+1) -- (\x+1,\y)
		};
		\draw \foreach \x in {1,...,6} \foreach \y in {1,...,6}{
			node[circle, draw, fill=white, inner sep=0pt, minimum width=7pt] (\x\y) at (\x,\y) {}
		};
		\draw \foreach \x in {(11),(12),(13),(14),(15),(51),(52),(53),(54),(55),(21),(31),(41),(25),(35),(45)}{
			\x node[circle, draw, fill=black, inner sep=0pt, minimum width=7pt] {}
		};
        \draw \foreach \x in {(33)}{
			\x node[circle, draw, fill=red, inner sep=0pt, minimum width=7pt] {}
		};
	\end{tikzpicture}
	\caption{The graph $P_6 \boxtimes P_6$ with the sets $D_2(2,2)$ and $D_2(3,3)$, respectively illustrated in black. The vertices $(2,2)$ and $(3,3)$ appear in red color.}\label{fig:P6P6-sets-D}\vspace*{-4mm}
\end{figure}
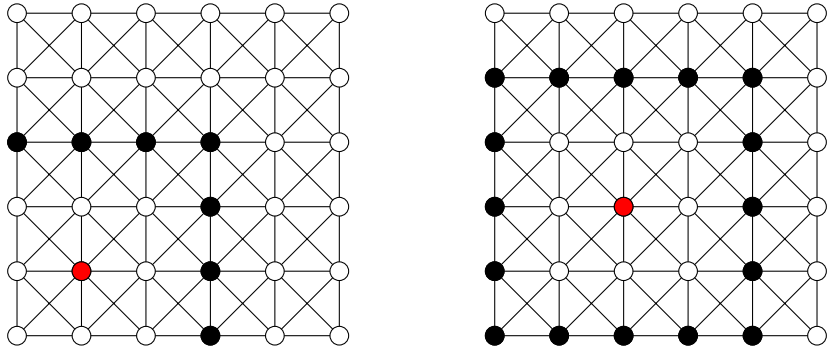

\begin{remark}
	We have $\dim_{ms} (P_2 \boxtimes P_2) = \dim_{ms} (P_3 \boxtimes P_3) = \infty$, since both of these graphs are of diameter at most 2, and such graphs have no multiset resolving sets~\cite{rino-2017}. As for the case $n=4$, the set $\{(1,1),(2,1),(4,1),(1,3),(2,3),(2,4)\}$ is a multiset resolving set of $P_4 \boxtimes P_4$. We have checked with an exhaustive computer search that no smaller multiset resolving sets exist for this graph, and thus $\dim_{ms} (P_4 \boxtimes P_4) = 6$
\end{remark}

We begin our exposition of the larger king grids with the two smallest cases, and further on proceed with the general case.\vspace*{-1mm}
\begin{proposition}\label{prop:P5P5}
We have $\dim_{ms} (P_5 \boxtimes P_5) = \dim_{ms} (P_6 \boxtimes P_6)= 4$.\vspace*{-2mm}
\end{proposition}

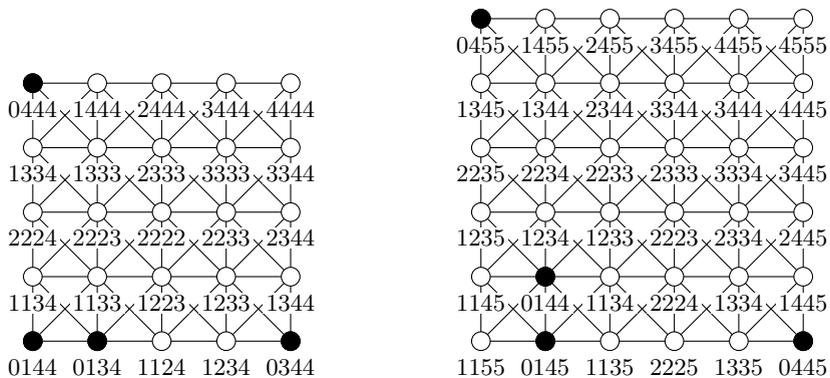
\begin{figure}[!b]
	\centering
	\begin{tikzpicture}[scale=0.85]
		\def\size{\footnotesize}
		\draw \foreach \x in {1,...,4} \foreach \y in {1,...,5} {
			(\x,\y) -- (\x+1,\y)
		};
		\draw \foreach \x in {1,...,5} \foreach \y in {1,...,4} {
			(\x,\y) -- (\x,\y+1)
		};
		\draw \foreach \x in {1,...,4} \foreach \y in {1,...,4} {
			(\x,\y) -- (\x+1,\y+1)
			(\x,\y+1) -- (\x+1,\y)
		};
		\draw \foreach \x in {1,...,5} \foreach \y in {1,...,5}{
			node[circle, draw, fill=white, inner sep=0pt, minimum width=7pt] (\x\y) at (\x,\y) {}
		};
		\draw \foreach \x in {(11),(21),(15),(51)}{
			\x node[circle, draw, fill=black, inner sep=0pt, minimum width=7pt] {}
		};
		\draw
		(11) node[rectangle, inner sep=1pt,fill=white,below=0.2 cm] {\size $0144$}
		(21) node[rectangle, inner sep=1pt,fill=white,below=0.2 cm] {\size $0134$}
		(31) node[rectangle, inner sep=1pt,fill=white,below=0.2 cm] {\size $1124$}
		(41) node[rectangle, inner sep=1pt,fill=white,below=0.2 cm] {\size $1234$}
		(51) node[rectangle, inner sep=1pt,fill=white,below=0.2 cm] {\size $0344$}
		(12) node[rectangle, inner sep=1pt,fill=white,below=0.2 cm] {\size $1134$}
		(22) node[rectangle, inner sep=1pt,fill=white,below=0.2 cm] {\size $1133$}
		(32) node[rectangle, inner sep=1pt,fill=white,below=0.2 cm] {\size $1223$}
		(42) node[rectangle, inner sep=1pt,fill=white,below=0.2 cm] {\size $1233$}
		(52) node[rectangle, inner sep=1pt,fill=white,below=0.2 cm] {\size $1344$}
		(13) node[rectangle, inner sep=1pt,fill=white,below=0.2 cm] {\size $2224$}
		(23) node[rectangle, inner sep=1pt,fill=white,below=0.2 cm] {\size $2223$}
		(33) node[rectangle, inner sep=1pt,fill=white,below=0.2 cm] {\size $2222$}
		(43) node[rectangle, inner sep=1pt,fill=white,below=0.2 cm] {\size $2233$}
		(53) node[rectangle, inner sep=1pt,fill=white,below=0.2 cm] {\size $2344$}
		(14) node[rectangle, inner sep=1pt,fill=white,below=0.2 cm] {\size $1334$}
		(24) node[rectangle, inner sep=1pt,fill=white,below=0.2 cm] {\size $1333$}
		(34) node[rectangle, inner sep=1pt,fill=white,below=0.2 cm] {\size $2333$}
		(44) node[rectangle, inner sep=1pt,fill=white,below=0.2 cm] {\size $3333$}
		(54) node[rectangle, inner sep=1pt,fill=white,below=0.2 cm] {\size $3344$}
		(15) node[rectangle, inner sep=1pt,fill=white,below=0.2 cm] {\size $0444$}
		(25) node[rectangle, inner sep=1pt,fill=white,below=0.2 cm] {\size $1444$}
		(35) node[rectangle, inner sep=1pt,fill=white,below=0.2 cm] {\size $2444$}
		(45) node[rectangle, inner sep=1pt,fill=white,below=0.2 cm] {\size $3444$}
		(55) node[rectangle, inner sep=1pt,fill=white,below=0.2 cm] {\size $4444$}
		;
	\end{tikzpicture}
\hspace*{1.6cm}
	\begin{tikzpicture}[scale=0.85]
		\def\size{\footnotesize}
		\draw \foreach \x in {1,...,5} \foreach \y in {1,...,6} {
			(\x,\y) -- (\x+1,\y)
		};
		\draw \foreach \x in {1,...,6} \foreach \y in {1,...,5} {
			(\x,\y) -- (\x,\y+1)
		};
		\draw \foreach \x in {1,...,5} \foreach \y in {1,...,5} {
			(\x,\y) -- (\x+1,\y+1)
			(\x,\y+1) -- (\x+1,\y)
		};
		\draw \foreach \x in {1,...,6} \foreach \y in {1,...,6}{
			node[circle, draw, fill=white, inner sep=0pt, minimum width=7pt] (\x\y) at (\x,\y) {}
		};
		\draw \foreach \x in {(21),(22),(16),(61)}{
			\x node[circle, draw, fill=black, inner sep=0pt, minimum width=7pt] {}
		};
		\draw
		(11) node[rectangle, inner sep=1pt,fill=white,below=0.2 cm] {\size $1155$}
		(21) node[rectangle, inner sep=1pt,fill=white,below=0.2 cm] {\size $0145$}
		(31) node[rectangle, inner sep=1pt,fill=white,below=0.2 cm] {\size $1135$}
		(41) node[rectangle, inner sep=1pt,fill=white,below=0.2 cm] {\size $2225$}
		(51) node[rectangle, inner sep=1pt,fill=white,below=0.2 cm] {\size $1335$}
        (61) node[rectangle, inner sep=1pt,fill=white,below=0.2 cm] {\size $0445$}
		(12) node[rectangle, inner sep=1pt,fill=white,below=0.2 cm] {\size $1145$}
		(22) node[rectangle, inner sep=1pt,fill=white,below=0.2 cm] {\size $0144$}
		(32) node[rectangle, inner sep=1pt,fill=white,below=0.2 cm] {\size $1134$}
		(42) node[rectangle, inner sep=1pt,fill=white,below=0.2 cm] {\size $2224$}
		(52) node[rectangle, inner sep=1pt,fill=white,below=0.2 cm] {\size $1334$}
		(62) node[rectangle, inner sep=1pt,fill=white,below=0.2 cm] {\size $1445$}
		(13) node[rectangle, inner sep=1pt,fill=white,below=0.2 cm] {\size $1235$}
		(23) node[rectangle, inner sep=1pt,fill=white,below=0.2 cm] {\size $1234$}
		(33) node[rectangle, inner sep=1pt,fill=white,below=0.2 cm] {\size $1233$}
		(43) node[rectangle, inner sep=1pt,fill=white,below=0.2 cm] {\size $2223$}
		(53) node[rectangle, inner sep=1pt,fill=white,below=0.2 cm] {\size $2334$}
		(63) node[rectangle, inner sep=1pt,fill=white,below=0.2 cm] {\size $2445$}
		(14) node[rectangle, inner sep=1pt,fill=white,below=0.2 cm] {\size $2235$}
		(24) node[rectangle, inner sep=1pt,fill=white,below=0.2 cm] {\size $2234$}
		(34) node[rectangle, inner sep=1pt,fill=white,below=0.2 cm] {\size $2233$}
		(44) node[rectangle, inner sep=1pt,fill=white,below=0.2 cm] {\size $2333$}
		(54) node[rectangle, inner sep=1pt,fill=white,below=0.2 cm] {\size $3334$}
		(64) node[rectangle, inner sep=1pt,fill=white,below=0.2 cm] {\size $3445$}
		(15) node[rectangle, inner sep=1pt,fill=white,below=0.2 cm] {\size $1345$}
		(25) node[rectangle, inner sep=1pt,fill=white,below=0.2 cm] {\size $1344$}
		(35) node[rectangle, inner sep=1pt,fill=white,below=0.2 cm] {\size $2344$}
		(45) node[rectangle, inner sep=1pt,fill=white,below=0.2 cm] {\size $3344$}
		(55) node[rectangle, inner sep=1pt,fill=white,below=0.2 cm] {\size $3444$}
		(65) node[rectangle, inner sep=1pt,fill=white,below=0.2 cm] {\size $4445$}
		(16) node[rectangle, inner sep=1pt,fill=white,below=0.2 cm] {\size $0455$}
		(26) node[rectangle, inner sep=1pt,fill=white,below=0.2 cm] {\size $1455$}
		(36) node[rectangle, inner sep=1pt,fill=white,below=0.2 cm] {\size $2455$}
		(46) node[rectangle, inner sep=1pt,fill=white,below=0.2 cm] {\size $3455$}
		(56) node[rectangle, inner sep=1pt,fill=white,below=0.2 cm] {\size $4455$}
		(66) node[rectangle, inner sep=1pt,fill=white,below=0.2 cm] {\size $4555$}
		;
	\end{tikzpicture}
	\caption{The graphs $P_5 \boxtimes P_5$ and $P_6 \boxtimes P_6$ with the sets $S$ and $S'$, respectively illustrated in black. The four digits below each vertex are the distances in the multiset representation sorted in ascending order.}\label{fig:P5P5}
\end{figure}

\begin{proof}
	It is known from \cite{rino-2017} and \cite{Chartrand-2021} that no graph has a multiset resolving set of cardinality $2$. Also, the only graph that has a multiset resolving set consisting of only one element is the path graph $P_n$. Thus, $\dim_{ms} (P_5 \boxtimes P_5) \geq 3$ and $\dim_{ms} (P_6 \boxtimes P_6) \geq 3$.
	We first prove that $\dim_{ms} (P_5 \boxtimes P_5) \neq 3$ by a simple counting argument. Suppose that $S$ is a multiset resolving set of $P_5 \boxtimes P_5$ such that $|S| = 3$. The maximum number of distinct multiset representations that do not contain $0$ is $\binom{6}{3} = 20$. Since $|V(G) \setminus S| = 22$, some vertices of $P_5 \boxtimes P_5$ have the same multiset representations, a contradiction. Thus, $\dim_{ms} (P_5 \boxtimes P_5)> 3$, and the first equality follows since the set $S = \{ (1,1), (2,1), (5,1), (1,5) \}$ is a multiset resolving set of $P_5 \boxtimes P_5$. The sets $S$ along with the multiset representations is illustrated in Figure~\ref{fig:P5P5}.

\medskip
On the other hand, observe that the set $S' = \{ (2,1),(2,2),(6,1),(1,6) \}$ is a multiset resolving set of $P_6 \boxtimes P_6$, as shown in Figure~\ref{fig:P5P5}, throughout the multiset representations of each vertex with respect to $S'$, being pairwise different. Thus $\dim_{ms} (P_6 \boxtimes P_6)\le 4$.

\smallskip
Now, in contrast to the case of $P_5 \boxtimes P_5$, to prove that $\dim_{ms} (P_6 \boxtimes P_6)\ne 3$, the counting argument used does not directly work. That is, if we suppose $S''$ is a multiset resolving set of $P_6 \boxtimes P_6$ such that $|S| = 3$, then the maximum number of distinct multiset representations that do not contain $0$ is $\binom{7}{3} = 35$, and $|V(P_6 \boxtimes P_6)\setminus S''|=33$. Thus, some extra arguments are required. To this end, assume there is a vertex $(\alpha,\beta)\in V(P_6 \boxtimes P_6)$ such that it has multiset representation $\msl 1,1,a\msr$ for some $a\in\{1,\dots,5\}$. Let $(i,j),(i',j')$ be two neighbors of $(\alpha,\beta)$ in $S''$ (note that $(i,j),(i',j')$ are at distance at most two). Hence, the third vertex $(i'',j'')$ of $S''$ must be in the set $D_a(\alpha,\beta)$, namely, $S''=\{(i,j),(i',j'),(i'',j'')\}$.

\medskip
It is now just a matter of checking all the possibilities that can occur between the two vertices $(i,j),(i',j')$ and the third vertex $(i'',j'')$, to observe that one can always find two vertices that have the same multiset representation with respect to $S''$. In order to avoid a lengthy and time consuming case analysis, we have simply checked this by computer. Thus, the multisets representations $\msl 1,1,a\msr$ are not possible for every $a\in\{1,\dots,5\}$ with respect to the set $S''$, and there are $5$ of them. But then this means we have a total amount of $\binom{7}{3} - 5= 30$ possible distinct multiset representations that do not contain $0$. However, there are $|V(P_6 \boxtimes P_6)\setminus S''|=33$ vertices, which is a contradiction. Therefore, $\dim_{ms} (P_6 \boxtimes P_6)\ne 3$, and the equality $\dim_{ms} (P_6 \boxtimes P_6)= 4$ follows.
\end{proof}

From now on, in order to facilitate the exposition, by a \emph{row} or a \emph{column} in $P_n \boxtimes P_n$ we mean the path induced by the vertices $(1,j),(2,j)\dots,(n,j)$ or $(i,1),(i,2)\dots,(i,n)$, respectively, for any $i,j\in \{1,\dots, n\}$.

\begin{theorem}
If $n \geq 7$ is an integer, then $3\le \dim_{ms}(P_n \boxtimes P_n) \le 4$.
\end{theorem}

\begin{proof}
The lower bound follows from the fact that any graph different from a path has multiset dimension at least $3$, or by using inequality \eqref{eq-dim-mdim}, since $\dim(P_n \boxtimes P_n)=3$ (see \cite{Barragan,Rodriguez}). To show the upper bound we use an induction procedure that separately works for odd and even values of $n$.

\medskip
\noindent
	\textbf{Case 1:} $n \geq 7$ is odd. Assume that the set $S_{n-2} = \{(1,1),(2,1),(n-2,1),(1,n-2)\}$ is a multiset resolving set of $P_{n-2} \boxtimes P_{n-2}$. This is true for $n=7$ according to Proposition~\ref{prop:P5P5}, which shows the base case of the induction process. We will show that the set $S_n = \{(1,1),(2,1),(n,1),(1,n)\}$ is a multiset resolving set of $P_n \boxtimes P_n$.

\medskip	
	Denote $F = \{ (x,y) \in V(P_n \boxtimes P_n) \, | \, x \in \{1,n\} \text{ or } y \in \{1,n\} \}$ and $I = V(P_n \boxtimes P_n) \setminus F$ (see Figure~\ref{fig:nodd}). Observe that each vertex in $F$ has $n-1$ in its multiset representation with respect to $S_n$, whereas the vertices in $I$ do not. Thus, vertices from $F$ and $I$ clearly have multiset representations distinct from one another.

\begin{figure}[!h]
		\centering
		\begin{subfigure}[b]{0.45\linewidth}
			\centering
			\begin{tikzpicture}[scale=.84]
				\def\size{\footnotesize}
				\def\n{7}
				\def\nn{6}
				\draw \foreach \x in {1,...,\nn} \foreach \y in {1,...,\n} {
					(\x,\y) -- (\x+1,\y)
				};
				\draw \foreach \x in {1,...,\n} \foreach \y in {1,...,\nn} {
					(\x,\y) -- (\x,\y+1)
				};
				\draw \foreach \x in {1,...,\nn} \foreach \y in {1,...,\nn} {
					(\x,\y) -- (\x+1,\y+1)
					(\x,\y+1) -- (\x+1,\y)
				};
				\draw \foreach \x in {1,...,\n} \foreach \y in {1,...,\n}{
					node[circle, draw, fill=white, inner sep=0pt, minimum width=7pt] (\x\y) at (\x,\y) {}
				};
				\draw \foreach \x in {(11),(21),(1\n),(\n1)}{
					\x node[circle, draw, fill=black, inner sep=0pt, minimum width=7pt] {}
				};
				\draw \foreach \x in {(22),(32),(26),(62)}{
					\x node[circle, draw, fill=gray!60, inner sep=0pt, minimum width=7pt] {}
				};
				\draw[thick,dashed] (22)+(-.3,-.3) -| +(4.3,4.3) -| cycle;
			\end{tikzpicture}
			\caption{$n=7$}\label{fig:nodd}
		\end{subfigure}
		\hfill
		\begin{subfigure}[b]{0.45\linewidth}
			\centering
	\hspace*{-12mm}\begin{tikzpicture}[scale=.84]
				\def\size{\footnotesize}
				\def\n{8}
				\def\nn{7}
				\draw \foreach \x in {1,...,\nn} \foreach \y in {1,...,\n} {
					(\x,\y) -- (\x+1,\y)
				};
				\draw \foreach \x in {1,...,\n} \foreach \y in {1,...,\nn} {
					(\x,\y) -- (\x,\y+1)
				};
				\draw \foreach \x in {1,...,\nn} \foreach \y in {1,...,\nn} {
					(\x,\y) -- (\x+1,\y+1)
					(\x,\y+1) -- (\x+1,\y)
				};
				\draw \foreach \x in {1,...,\n} \foreach \y in {1,...,\n}{
					node[circle, draw, fill=white, inner sep=0pt, minimum width=7pt] (\x\y) at (\x,\y) {}
				};
				\draw \foreach \x in {(22),(32),(28),(82)}{
					\x node[circle, draw, fill=black, inner sep=0pt, minimum width=7pt] {}
				};
				\draw[thick,dashed] (22)+(-.3,-.3) -| +(6.3,6.3) -| cycle;
			\end{tikzpicture}
			\caption{$n=8$}\label{fig:neven}
		\end{subfigure}	\vspace*{-3mm}
		\caption{Vertices within the dashed line are the vertices in $I$, and the vertices outside the dashed line are the vertices in $F$. The black vertices are the elements of the set $S_n$.}\vspace*{-2mm}
	\end{figure}
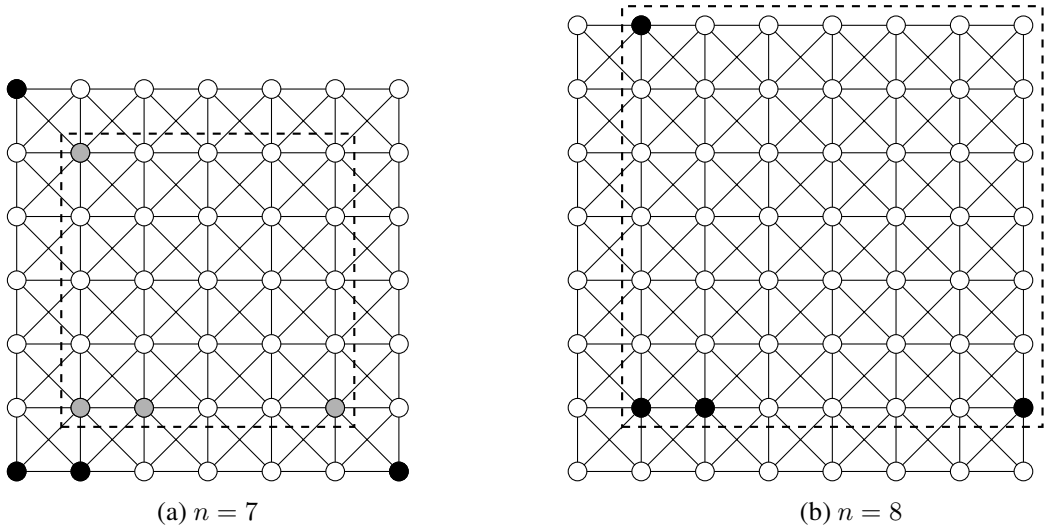

\medskip	
	Consider the vertices in $I$. The graph $P_n \boxtimes P_n$ can be viewed as a graph constructed from $P_{n-2} \boxtimes P_{n-2}$ by adding an additional row or column of vertices to all four sides of the graph. The set $S_n$ can then be obtained by moving the elements of $S_{n-2}$ diagonally away from the middle. Thus, for each $(v_1,v_2) \in I \setminus \{(2,2),(3,2)\}$, we have
	\[ \msrepr((v_1,v_2) | S_n) = \msrepr((v_1-1,v_2-1) | S_{n-2}) + 1. \]
	Since $S_{n-2}$ is a multiset resolving set of $P_{n-2} \boxtimes P_{n-2}$, all vertices in $I \setminus \{(2,2),(3,2)\}$ have distinct multiset representations with respect to $S_n$ in $P_n \boxtimes P_n$. Since $\msrepr((2,2) | S_n) = \msl 1,1,n-2,n-2 \msr$ and $\msrepr((3,2) | S_n) = \msl 1,2,n-3,n-2 \msr$, the vertices $(2,2)$ and $(3,2)$ have distinct multiset representations with respect to each other. Moreover, the vertex $(2,2)$ is the only vertex in $I$ that is adjacent to two elements of the set $S_2$. The vertex $(3,2)$ is adjacent to one element of the set $S_n$, and so are two other vertices in $I$. However, neither of these vertices has the distance $2$ in their multiset representations, which the vertex $(3,2)$ does have. Thus, all vertices in $I$ have pairwise distinct multiset representations with respect to $S_n$.
	
\medskip
	Consider then the vertices in $F$. The multiset representations of these vertices are as follows.
	\begin{itemize}
		\item For vertices on the top row, i.e. $(i,n)$ where $i \in \{1, \ldots , n \}$, we have $\msrepr((i,n) | S_n) = \msl i-1, n-1,n-1,n-1 \msr$.
		\item For vertices on the right column, excluding the top corner, i.e. $(n,i)$ where $i \in \{1, \ldots, n-1\}$, we have $\msrepr((n,i) | S_n) = \msl i-1,n-2,n-1,n-1 \msr$.
		\item For vertices on the left column, excluding top and bottom corners, i.e. $(1,i)$ where $i \in \{2, \ldots , n-1 \}$, we have $\msrepr((1,i) | S_n) = \msl i-1,i-1,n-i-2,n-1 \msr$.
		\item For vertices on the bottom row, excluding both corners, i.e. $(i,1)$, where $i \in \{2, \ldots , n-1 \}$, we have $\msrepr((i,1) | S_n) = \msl i-1,i-2,n-i-2,n-1 \msr$.
		\item For the bottom left corner $(1,1)$, we have $\msrepr((1,1) | S_n) = \msl 0,1,n-1,n-1 \msr$.\vspace*{1mm}
	\end{itemize}
	Since the vertices in the top row are the only ones to have (at least) three $(n-1)$'s in their multiset representations, the vertices on the top row have distinct representations from the other vertices of $F$. Moreover, the multiset representations of the vertices on the top row are pairwise distinct. Similarly, from the vertices left to be considered, the vertices in the right column are the only vertices that have exactly two $(n-1)$'s in their representations (other than $(1,1)$, but that has a distinct multiset representation due to the 0). Thus it is clear that the vertices in the right column have distinct multiset representations with respect to each other and other vertices in $F$. The vertices in the left column have pairwise distinct representations with respect to each other, and the same holds also for vertices in the bottom row. The only thing left to show is that two vertices, one in the left column and the other in the bottom row, cannot have the same multiset representations.
	
\medskip
	To that end, suppose to the contrary that $(i,1)$ and $(1,j)$ have the same multiset representation for some $i,j \in \{2, \ldots , n-1 \}$. Since $j-1$ appears twice in the multiset representation of $(1,j)$, we must have $i-1 = n-i-2$ or $i-2 = n-i-2$. Since $n$ is odd, we have $i-1 = n-i-2$ and $i-1 = \frac{n-3}{2}$. Now, $\msrepr((i,1) | S_n) = \msl \frac{n-3}{2}, \frac{n-3}{2}, \frac{n-3}{2}-1,n-1 \msr$. This implies that $j-1 = \frac{n-3}{2}$, but now $\msrepr((1,j) | S_n) = \msl \frac{n-3}{2},\frac{n-3}{2},\frac{n-3}{2},n-1 \msr \neq \msrepr((i,1)|S_n)$, a contradiction. Thus, all vertices of $F$ have pairwise distinct multiset representations with respect to $S_n$.
	
As a consequence, we obtain that $S_n$ is a multiset resolving set as claimed, and so, $\dim_{ms}(P_n \boxtimes P_n) \le 4$ in this case.
	
\medskip
\noindent
\textbf{Case 2:} $n\geq 8$ is even. Assume that the set $S_{n-1} = \{(1,1),(2,1),(n-1,1),(1,n-1)\}$ is a multiset resolving set of $P_{n-1} \boxtimes P_{n-1}$. If $n=8$, then by the Case 1 we know that $S_{7}$ is a multiset resolving set of $P_{7} \boxtimes P_{7}$, which shows the base case. We will show that the set $S_n = \{(2,2),(2,3),(n,2),(2,n)\}$ is a multiset resolving set of $P_n \boxtimes P_n$.

\medskip	
	We now denote $F = \{ (x,y) \in V(P_n \boxtimes P_n) \, | \, x = 1 \text{ or } y = 1 \}$ and $I = V(P_n \boxtimes P_n) \setminus F$ (see Figure~\ref{fig:neven}). The vertices in $I$ clearly have pairwise distinct multiset representations as $\msrepr(v | S_n) = \msrepr(v | S_{n-1})$ for all $v \in I$, and $S_{n-1}$ is a multiset resolving set of $P_{n-1} \boxtimes P_{n-1}$. Also, all the vertices in $F$ have $n-1$ in their multiset representations with respect to $S_n$, whereas vertices in $I$ do not. Thus, the multiset representation of a vertex in $F$ is always distinct from that of a vertex of $I$.

\medskip	
	We will show next that the multiset representations of vertices of $F$ are pairwise distinct.
	To that end, we first consider the vertices that are adjacent to some element of $S_n$. Their multiset representations are the following:
	\begin{align*}
		\msrepr((1,1)|S_n) &= \msl 1,2, n-1,n-1 \msr, & \msrepr((1,n)|S_n) &= \msl 1,n-2,n-2,n-1 \msr, \\
		\msrepr((1,2)|S_n) &= \msl 1,2, n-2,n-1 \msr, & \msrepr((1,n-1)|S_n) &= \msl 1,n-3,n-3,n-1 \msr, \\
		\msrepr((1,3)|S_n) &= \msl 1,2, n-3,n-1 \msr, & \msrepr((n,1)|S_n) &= \msl 1,n-3,n-2,n-1 \msr, \\
		\msrepr((4,1)|S_n) &= \msl 1,2, n-4,n-1 \msr, & \msrepr((n-1,1)|S_n) &= \msl 1,n-4,n-3,n-1 \msr, \\
		\msrepr((2,1)|S_n) &= \msl 1,1, n-2,n-1 \msr, &
		\msrepr((3,1)|S_n) &= \msl 1,1, n-3,n-1 \msr.
	\end{align*}
	Since $n \geq 8$, we have $n-4 \neq 2$ and all these multiset representations are pairwise distinct. Moreover, since these vertices are the only vertices in $F$ that have the distance 1 in their multiset representations, these multiset representations are distinct from the multiset representation of other vertices in $F$ as well. Let us then consider the rest of the vertices in $F$. The vertices $(i,1)$ where $i \in \{5,\ldots,n-2\}$ have multiset representations of the form $\msl i-2,i-3,n-i-2,n-1 \msr$, and these representations are  pairwise distinct. Similarly, the vertices $(1,j)$ where $j \in \{4, \ldots , n-2\}$ have multiset representations of the form $\msl j-2, j-2,n-j-2,n-1 \msr$, and these representations are clearly pairwise distinct.
	Suppose then that $(i,1)$ and $(1,j)$ have the same multiset representation for some $i \in \{5, \ldots , n-2\}$ and $j \in \{4,\ldots,n-2\}$. As in the proof for odd $n$, the distance $j-2$ appears twice in the multiset representation of $(1,j)$. This implies that $i-2 = n-i-2$ or $i-3=n-i-2$. As $n$ is even, we have $i-2 = n-i-2$, and thus $i-2 = \frac{n-4}{2}$. Now, $\msrepr((i,1)|S_n) = \msl \frac{n-4}{2},\frac{n-4}{2},\frac{n-4}{2}-1,n-1 \msr$. This implies that $j-2 = \frac{n-4}{2}$. However, now $\msrepr((1,j)|S_n) = \msl \frac{n-4}{2},\frac{n-4}{2},\frac{n-4}{2},n-1 \msr \neq \msrepr((i,1)|S_n)$, a contradiction.
	Thus, all vertices in $F$ have pairwise distinct multiset representations with respect to~$S_n$.

Therefore, we again conclude that $S_n$ is a multiset resolving set in this situation, and so, $\dim_{ms}(P_n \boxtimes P_n) \le 4$ follows as well, which completes the proof.
\end{proof}

\section{Strong products involving a complete graph}
\label{sec:strong-prod}

Several graphs with infinite multiset dimension (or equivalently that are not ID-graphs) are already known from the seminal works \cite{Chartrand-2021,rino-2017}, in their corresponding terminologies. For instance, it is known from these mentioned works that for a given graph $G$ of diameter two, $\dim_{ms}(G\boxtimes K_n)<\infty$ if and only $G$ is $P_3$. From this result we can identify a lot of interesting and non-trivial families of graphs like for instance the join of two non-complete graphs, the Cartesian or direct product of two complete graphs, and the Kneser graph $K(n, 2)$, among others, having infinite multiset dimension.

\medskip
We next describe some other graphs with diameter larger than two that have infinite multiset dimension. Specifically, we consider the case of strong product graphs $G\boxtimes H$ when $H$ is a complete graph. We first need the following definition from \cite{Klavzar-2023}. A graph $G$ is called a \emph{multiset distance irregular graph} if for any two vertices $x,y\in V(G)$ it follows that $\msrepr(x|V(G))\ne \msrepr(y|V(G))$. A example of a multiset distance irregular graph is for instance a tree obtained from a star $S_n$ with leaves $v_0,\dots v_{n-1}$ ($n\ge 3$) and center $x$, by subdividing $i$ times the edge $xv_i$ for every $i\in\{0\dots,n-1\}$.

\begin{theorem}
\label{th:G-K_2}
Let $G$ be a graph and let $n\ge 2$ be an integer. We have $\dim_{ms}(G\boxtimes K_n)<\infty$ if and only if $n=2$ and $G$ is a multiset distance irregular graph.
\end{theorem}

\begin{proof}
Notice that the $n$ vertices of each copy of $K_n$ are twins, that is, they have the same closed neighborhood. Hence, if $S$ is a multiset resolving set of $G\boxtimes K_n$, then at least $n-1$ vertices from each copy of $K_n$ must be in $S$. If $n\ge 3$, then from each copy of $K_n$ there are at least two vertices in $S$. But then, these two vertices have the same multiset representation, which is not possible. Consequently, we deduce that $G\boxtimes K_n$ does not have multiset resolving sets when $n\ge 3$ and so, $\dim_{ms}(G\boxtimes K_n)=\infty$ in this case.

Assume next that $n=2$. By the same reasons as above, at least one vertex from each copy of $K_2$ must be in $S$. Thus $\dim_{ms}(G\boxtimes K_2)\ge |V(G)|$. Now, assume $G$ is multiset distance irregular, and consider a set $X$ of vertices containing one whole copy of $G$. Since the multisets representations of any two vertices $x,y\in V(G)$ satisfy that $\msrepr(x|V(G))\ne \msrepr(y|V(G))$, we deduce that $X$ is a multiset resolving set, and so $\dim_{ms}(G\boxtimes K_n)\le |V(G)|$, which gives the equality $\dim_{ms}(G\boxtimes K_n)= |V(G)|< \infty$.

\medskip
Assume now that $\dim_{ms}(G\boxtimes K_n)<\infty$. Clearly, $n=2$, for otherwise we get a contradiction. Also, we readily see that every multiset resolving set has nonempty intersection with every copy of $K_2$ in $G\boxtimes K_2$ and that $\dim_{ms}(G\boxtimes K_2)\ge |V(G)|$. Let $X'$ be a multiset basis of $G\boxtimes K_2$. If $X'$ contains both vertices of one copy of $K_2$, then these two vertices have the same multiset representation as they are twins. Thus, $X'$ contains exactly one vertex from each copy of $K_2$. If $G$ is not multiset distance irregular, then there are two vertices $x,y\in V(G)$ such that $\msrepr(x|V(G))=\msrepr(y|V(G))$. Now, the vertices in the copies of $K_2$ corresponding to $x$ and $y$ that are not in $X'$ have the same multiset representation, a contradiction. Therefore, $G$ is multiset distance irregular, and the proof is completed.
\end{proof}

\begin{corollary}
For any  multiset distance irregular graph $G$, $\dim_{ms}(G\boxtimes K_2)=|V(G)|$.
\end{corollary}

\section{Concluding remarks}

This work firstly shows that two metric parameters represent the same in graph theory, that is, multiset dimension and ID-colorings are the same. We have also considered some computational and combinatorial problems on this parameter. As a consequence of the study a number of possible future research lines have been detected. We next remark a few that could be of interest from our humble opinion.

\begin{itemize}
  \item We have proved that finding the multiset dimension of graphs is in general NP-hard. However, not much is known on special classes of graphs. Does finding such parameters remain NP-hard even when restricted to trees, chordal graphs or planar graphs?
  \item The multiset dimension of the king grid has been bounded above by 4. Is it true that $\dim_{ms}(P_n \boxtimes P_n) = 4$ for any $n\ge 7$? In addition, this result gives a first step into considering the multiset dimension of the strong product graphs in general. Moreover, the study of such parameter in some other related graph products is worthwhile as well.
  \item We have related the multiset dimension of graphs with multiset distance irregular graphs. Does such graphs play a significant role in some investigations on the multiset dimension of general graphs or at least in some product related structures?
  \item  The notion of identification spectrum of an ID-graph $G$ was introduced in \cite{Chartrand-2021}, which can be understood as the set of positive integers $r$ for which there is a multiset resolving set of cardinality $r$. In this sense, since we have proved that $\dim_{ms}(P_n \boxtimes P_n) \le 4$ (equivalently $P_n \boxtimes P_n$ is an ID-graph) for any $n\ge 7$, we wonder which is the identification spectrum of $P_n \boxtimes P_n$.
  \item Characterize all the graphs $G$ for which $\dim(G)=\dim_{ms}(G)$.
\end{itemize}

\subsection*{Acknowledgements}

Anni Hakanen has been supported by the Jenny and Antti Wihuri Foundation, and also partially supported by the ANR project GRALMECO (ANR-21-CE48-0004-01) and Academy of Finland grant number 338797. Ismael G. Yero has been partially supported by the Spanish Ministry of Science and Innovation through the grant PID2019-105824GB-I00.

\section*{Declaration of interests}

The authors declare that they have no known competing financial interests or personal relationships that could have appeared to influence the work reported in this paper.

\section*{Data availability}

Our manuscript has no associated data.

\end{document}